\documentclass[a4paper,11pt]{article}
\usepackage{fullpage}
\usepackage{indentfirst}
\usepackage{amsmath}
\usepackage{amsthm}
\usepackage{amssymb}
\usepackage{amsfonts}
\usepackage{fancybox}
\usepackage{fancyvrb}
\usepackage{color}
\usepackage{makeidx}
\usepackage{graphicx}
\usepackage{array}
\usepackage{cite}
\usepackage{url}
\usepackage{enumerate}
\usepackage{abstract}
\usepackage{subcaption}
\usepackage{tikz}
\usepackage[ruled,vlined]{algorithm2e}
\SetAlgorithmName{Auction}{Auction}{Auction}

\newtheorem{theorem}{Theorem}[section]

\newtheorem{lemma}[theorem]{Lemma}

\newtheorem{exam}[theorem]{Example}
\newcommand{\be}{{\mathrm{E}}}
\newcommand{\bv}{{\mathrm{Var}}}
\DeclareMathOperator*{\argmax}{\arg\max}

\newcommand{\CommentS}[1]{}

\title{Improved Efficiency Guarantees in Auctions with Budgets}
\date{}
\begin{document}
\author{
	Pinyan Lu\thanks{Microsoft Research. {\tt pinyanl@microsoft.com}}
\and
    Tao Xiao\thanks{Shanghai Jiao Tong University. {\tt taoxiao1992@gmail.com}}
}
\maketitle
\begin{abstract}
We study the efficiency guarantees in the simple auction environment where the auctioneer has one unit of divisible good to be distributed among a number of budget constrained agents. With budget constraints, the social welfare cannot be approximated by a better factor than the number of agents by any truthful mechanism. Thus, we follow a recent work by Dobzinski and Leme \cite{DBLP:conf/icalp/DobzinskiL14} to approximate the liquid welfare, which is the welfare of the agents each capped by her/his own budget. We design a new truthful auction with an approximation ratio of $\frac{\sqrt{5}+1}{2} \approx 1.618$, improving the best previous ratio of $2$ when the budgets for agents are public knowledge and their valuation is linear (additive). In private budget setting, we propose the first constant approximation auction with approximation ratio of $34$. Moreover, this auction works for any valuation function. Previously, only $O(\log n)$ approximation was known for linear and decreasing marginal (concave) valuations, and $O(\log^2 n)$ approximation was known for sub-additive valuations.
\end{abstract}


\section{Introduction}
We consider a simple auction environment: the auctioneer has certain amount of divisible good to be distributed among a number of $n$ agents. Since the good is divisible, without loss of generality, we can always assume that it is of one unit. Each agent $i\in [n]$ has a valuation function $v_i(\cdot)$ for the good (willingness-to-pay) and a budget $B_i$ to indicate the maximum amount of money she/he is able to pay to the auctioneer (ability-to-pay). We always assume that the valuation $v_i(\cdot)$ is private information for agent $i$ and we shall study both public budget model where the budgets are public knowledge to the auctioneer and private budget model where budget $B_i$ is also private information for agent $i$. Upon receiving the bids, the auction allocates $x_i\geq 0$ unit of the good to agent $i$ and charge her/him $p_i\geq 0$ amount of money. Then the utility of agent $i$ is $v_i(x_i) -p_i$ if $p_i\leq B_i$; otherwise her/his utility is $-\infty$ since she/he does not have enough money to pay. We call an auction  \emph{truthful} or \emph{incentive compatible} if it is always a dominant strategy for every bidder $i$ to submit her/his true private information.  We say a randomized auction is \emph{universally truthful} if it is a probabilistic distribution over deterministic truthful auctions. The auction is budget feasible if we always have $p_i\leq B_i$ for a truth-telling agent $i$.

If there is no budget constraint, the remarkable VCG auction~\cite{Vickrey1961,Clarke1971,Groves1973} is a truthful auction to achieve optimal social welfare. However, budget constraints for agents are very common in real life. For high value items such as spectrum, this is due to the ability-to-pay: an agent who values the item very high may not have enough money to pay it; even for relatively low value items such as key words auction for search engine, budget is also the first thing to concern for advertisers, since the volume for the auction could be very large, and a budget is used for risk control. The existence of budget brings in a huge challenge to the design of auctions and even theoretical impossibility results especially when social warfare is the objective. In particular, no truthful auction can approximate social welfare by a better factor than the number of agents even with publicly known budget constraints and linear valuation functions. The main reason is that we cannot truthfully allocate a significant amount of good to agents with very high values but small budgets.

To give a more realistic benchmark for social welfare, a new notion called \emph{liquid welfare} was proposed by Dobzinski and Leme~\cite{DBLP:conf/icalp/DobzinskiL14} as an alternative quantifiable measure for social efficiency. It is defined to be
 $$\overline{W}(\mathbf{x}) = \sum_i \min\{v_i(x_i), B_i\}.$$
 Basically, each agent's utility is capped by her/his budget. Therefore, an agent with high value but small  budget cannot contribute much to liquid welfare. This is a reasonable measure as argued in the paper~\cite{DBLP:conf/icalp/DobzinskiL14}: ``efficiency should be measured only with respect to the funds available to the bidder at the time of the auction, and not the additional liquidity he might gain after receiving the goods in the auction''. This is also the maximum amount of revenue an omniscient seller would be able to extract from a certain instance.  More justification for this measure can be found in the paper~\cite{DBLP:conf/icalp/DobzinskiL14}.

 With respect to this optimal liquid welfare objective, their paper gave two truthful auctions both with approximation ratio of $2$ in the public budget model with linear (additive) valuation functions and proved a lower bound of $\frac{4}{3}$ in this same setting. They explicitly asked whether one can have a truthful auction that provides an approximation ratio better than $2$ in this simple setting.
 For the more challenging private budget model, they provided an $O(\log n)$ approximation truthful auction for linear and decreasing marginal (concave) valuations, and an $O(\log^2 n)$ approximation auction for sub-additive valuations. The main open question is whether a constant approximation exists or not. This was not known even for simple linear valuation functions.

\subsection{Our Results and Techniques}
In this paper, we answer both of their open questions affirmatively. For the public budget setting and linear valuations, we design a new truthful auction with an approximation ratio of $\varphi=\frac{\sqrt{5}+1}{2}\approx 1.618$, where $\varphi$ is the golden ratio (i.e. the positive solution for the equation $t^2=t+1$). For the private budget setting, we design the first constant approximation auction with an approximation ratio of $34$. More importantly, our auction works for all valuation functions, not necessary linear, concave or sub-additive. This is a rather surprising result and this generality makes the auction applicable in many different scenarios.

Our design techniques are also new. For the $2$ approximation auction proposed in~\cite{DBLP:conf/icalp/DobzinskiL14}, the rough idea is to use a uniform market clearing price to sell the item to agents. Their ratio of $2$ is tight for their mechanism even for two agents. The bad case happens when one agent has very high value but limited budget while the other agent has a relatively lower value but enough budget. In the optimal allocation, the first agent gets very little share of the good, but this cannot be archived by a uniform pricing scheme. The high level idea of our mechanism is that  an agent can pay certain uniform price per unit but only use up certain fraction of her/his budget. In order to use up more of her/his budget, she/he needs to pay higher price per unit. By this mechanism, an agent with high value but limited budget will still use up all her/his budget but get less share of the good.

Our above mechanism crucially uses the fact that the auctioneer knows the budget for each agent. For the private budget setting, we go back to the uniform pricing scheme.
However, we do not know how to compute a good global  uniform price truthfully in private budget setting.
%
 To overcome this, we make use of random sampling, one of the most powerful techniques in truthful mechanism design~\cite{Goldberg2006,BeiCGL12,GravinL13}. We randomly divide the agents into two groups, compute the optimal liquid welfare for one group and use this as a guide to charge agents in the other group.

In order to make this random sampling auction work, the contribution in an optimal solution from different groups should be relatively balanced. In particular, if most of the contribution is from one single agent, random sampling does not work. Therefore, random sampling is usually combined with a Vickrey auction~\cite{BeiCGL12,ChenGL13} which works well in this unbalanced case. We also combine a Vickrey auction here for the modified valuation $\min\{v_i(1), B_i\}$. This Vickrey auction was also mentioned in~\cite{DBLP:conf/icalp/DobzinskiL14} and was claimed to be truthful there. However, we notice that there is a subtle issue due to budget constraint and tie-breaking which makes the auction not truthful. To overcome this, we modify the Vickrey auction in which the winner (with highest value) need to pay a bit higher than the second highest value. In the case that the two highest  values of the agents are very close to each other, the auction simply refuse to sell the item. We also design another version of modified Vickrey auction which works well when the two highest values are very close to each other. We think that this observation of untruthfulness and these modifications of Vickrey auction are of independent interest.

\subsection{Related Work}
Due to its practical relevance, many theoretical investigations have been devoted to analyzing auctions for budget constrained agents, especially in direction of optimal auction design which tries to maximize the revenue for the auctioneer~\cite{BorgsCIMS05,ChawlaMM11,FeldmanFLS12,DevanurHH13}.
For social efficiency, a number of previous works focus on the solution concept of Pareto Efficiency, which exist for the public budget model but not for private budget model~\cite{DobzinskiLN12,FiatLSS11}.

Similar alternative quantifiable measures for efficiency for budget constrained agents were also studied in~\cite{DevanurHH13,SyrgkanisT13} but for different solution concepts.

Another related topic is to study budget feasible mechanism design for reversal auction where the budget constrained buyer is the auctioneer rather than a bidder. This model was first proposed and studied by Singer~\cite{Singer10}. Since then, several improvements have been obtained~\cite{ChenGL11,DobzinskiPS11,BeiCGL12}.



\section{Public Budgets}
In this section, we consider the setting that agents' budgets are public information to the auctioneer, and the valuation function for each agent is linear. To simplify the notations, in this section we will use $v_i$ to denote value per unit for agent $i$ and thus $v_i(x_i) = v_i x_i$. Without loss of generality, we assume that there are $n$ agents with values $v_1 \geq \ldots \geq v_n$ and corresponding budgets $B_1, \ldots, B_n$. Let $\varphi = \frac{\sqrt{5}+1}{2}$ which is the golden ratio (i.e. the positive solution for the equation $t^2=t+1$).

For public budget and linear valuations model, it becomes a single dimensional parameter mechanism design problem with parameter $\mathbf{v} = (v_1, v_2, \ldots, v_n)$, thus an auction can be characterized by allocation rule $\mathbf{x} : \mathbb{R}^n_+ \rightarrow \mathbb{R}^n_+$ and payment rule $\mathbf{p}: \mathbb{R}^n_+ \rightarrow \mathbb{R}^n_+$ that maps $\mathbf{v}$ to a vector of allocations $\mathbf{x}(\mathbf{v})$ and a vector of payments $\mathbf{p}(\mathbf{v})$. We present the Myerson's Lemma \cite{myerson1981optimal}, which is a powerful tool in these settings.

\begin{lemma}
\label{lemmyerson}
A deterministic mechanism, with allocation and payment rule $\mathbf{x}$,$\mathbf{p}$
respectively, is truthful if and only if for each bidder $i$ and each $v_{-i}$, the following
conditions hold:
\begin{enumerate}
\item Monotone Allocation: $x_i(v_i, v_{-i})\leq x_i(v'_i, v_{-i})$ for all $v'_i \geq v_i$;
\item The payments are such that: $p_i(v_i, v_{-i})= v_i \cdot x_i(v_i,v_{-i})- \int_0^{v_i}x_i(u, v_{-i})du$.
\end{enumerate}
\end{lemma}

Our new auction for public budget model is presented in Auction~\ref{GSUPA}. Here we assume that $v_{n+1} = 0$ if occurs.
\IncMargin{1em}
\begin{algorithm}[h]
\SetKwInOut{Input}{input}\SetKwInOut{Output}{output}
\Input{$n $ agents with valuations $v_1\geq \ldots \geq v_n$ and corresponding budgets $B_1, \ldots, B_n$}
\Output{An allocation $(x_1, \ldots, x_n)$ and corresponding payments $(p_1, \ldots, p_n)$}
\BlankLine
\Begin{
Let $k\in [n]$ be the maximum integer s.t. $\frac{1}{\varphi} \sum_{j=1}^k B_j \leq v_k$\;
\If{$\frac{1}{\varphi}\sum_{j=1}^k B_j\geq v_{k+1}$}{
	\For{$i = 1$  \emph{\KwTo} $k$}{
		$\hat{v}_i \leftarrow \frac{v_i}{\sum_{j=1}^k B_j}$\;
		$x_i \leftarrow \frac{B_i}{\sum_{j=1}^k B_j}\min\{\hat{v}_i,1\}$\;
	}
	\For{$i = k+1$  \emph{\KwTo} $n$}{
		$x_i \leftarrow 0$\;
	}
	}
	\Else
	{
	\For{$i = 1$  \emph{\KwTo} $k$}{
		$\hat{v}_i \leftarrow \frac{ v_i}{\varphi v_{k+1}}$\;
		$x_i \leftarrow \frac{B_i}{\varphi v_{k+1}}\min\{\hat{v}_i,1\}$\;
	}
	$x_{k+1} \leftarrow \frac{1}{\varphi} - \sum_{i = 1}^k \frac{ B_i}{\varphi^2 v_{k+1}}$\;
	\For{$i = k+2$  \emph{\KwTo} $n$}{
		$x_i \leftarrow 0$\;
	}
	}
	\For{$i=1$ \emph{\KwTo} $n$}{
	$p_i \leftarrow v_i \cdot x_i(v_i,v_{-i})- \int_0^{v_i}x_i(u, v_{-i})du$;\quad\quad\quad//   Myerson's Payment Rule\	

	}
}
\caption{Auction for Public Budgets}\label{GSUPA}
\end{algorithm}
\DecMargin{1em}

Firstly, we verify that this is indeed a well-defined auction, namely the total amount of good it allocates does not exceed one unit.

 If $\frac{1}{\varphi}\sum_{j=1}^k B_j\geq v_{k+1}$,
\[\sum_{i=1}^n x_i = \sum_{i=1}^k x_i = \sum_{i=1}^k \frac{B_i}{\sum_{j=1}^k B_j}\min\{\hat{v}_i,1\}
				\leq \sum_{i=1}^k \frac{B_i}{\sum_{j=1}^k B_j} = 1.\]

 If $\frac{1}{\varphi}\sum_{j=1}^k B_j< v_{k+1}$,
\begin{align*}
\sum_{i=1}^n x_i = \sum_{i=1}^{k} x_i + x_{k+1} &= \sum_{i=1}^k \frac{  B_i}{\varphi v_{k+1}}\min\{\hat{v}_i,1\} +  \frac{1}{\varphi} - \sum_{i = 1}^k \frac{B_i}{\varphi^2 v_{k+1}} \\& \leq \sum_{i=1}^k \frac{ B_i}{\varphi v_{k+1}} + 1 - \sum_{i = 1}^k \frac{B_i}{\varphi v_{k+1}} = 1.
\end{align*}

\begin{theorem}
\label{thmpubmain}
For public budget model and linear valuations, Auction~\ref{GSUPA} is a truthful, budget feasible mechanism
with approximation ratio of at most $\varphi$ for liquid welfare.
\end{theorem}

These properties shall be proved in the following two subsections. The following notations are used in the whole section. Let $k$ be as defined in Auction \ref{GSUPA},  $p_0 = \max\{\sum_{i=1}^kB_i, \varphi v_{k+1}\}$, and  $k_1$ be the maximum integer s.t. $v_{k_1} \geq p_0$. For agent $i \in [n]$ let $\hat{v}_i = \frac{v_i}{p_0}$, which is as defined in Auction \ref{GSUPA}.
We call instances with $\frac{1}{\varphi}\sum_{j=1}^k B_j\geq v_{k+1}$ of case I and instances with $\frac{1}{\varphi}\sum_{j=1}^k B_j< v_{k+1}$ of case II. In most of our analysis, we distinguish these two cases and prove them separately. We have the following facts by the rule of our auction:
\begin{itemize}
  \item In case I, $v_1\geq \cdots \geq v_{k_1} \geq p_0 \geq v_{k_1+1} \geq  \cdots \geq v_{k} \geq \frac{p_0}{\varphi}  \geq v_{k+1} \geq \cdots \geq v_{n}$;
  \item In case II, $v_1\geq \cdots \geq v_{k_1} \geq p_0 \geq v_{k_1+1} \geq  \cdots \geq v_{k+1} = \frac{p_0}{\varphi} \geq v_{k+2} \geq \cdots \geq v_{n}$.
\end{itemize}
Thus $\hat{v}_i\geq 1$ for $i=1,2,\ldots, k_1$ and $\hat{v}_i\in [\frac{1}{\varphi}, 1)$ for $i=k_1+1, k_1+2,\ldots, k$.

\subsection{Truthfulness and Budget Feasibility}
By Myerson's Lemma, we only need to verify that the allocation function in our auction is monotone as our payment is already determined by Myerson's integration.
\begin{lemma}(Monotonicity)
\label{lemmono}
The allocation function in Auction~\ref{GSUPA} is monotone, i.e., $v_i \rightarrow x_i(v_i, v_{-i})$ is non-decreasing.
\end{lemma}
\begin{proof}
For case I, only the first $k$ agents get non-zero unit of the item, thus we only need to prove that for these agents, one's share is non-decreasing if one increases her/his bid. This is obvious since allocation $x_i = \frac{B_i}{\sum_{j=1}^k B_j}\min\{\hat{v}_i,1\} = \frac{B_i}{\sum_{j=1}^k B_j}\min\{\frac{v_i}{\sum_{j=1}^k B_j},1\} $ of agent $i$  is a monotone non-decreasing function in $v_i$.

Now we assume that we are in case II where only the first $k+1$ agents get non-zero unit of the item.
 By the same argument as above, the first $k$ agents get no less unit of the item if she/he increases her/his bid. We prove that this also holds for the $(k+1)$-th agent.

 For agent $k+1$, as she/he continues to increase her/his value and keep the $(k+1)$-th place, her/his allocation will increase, since it is $x_{k+1} = \frac{1}{\varphi} - \sum_{j=1}^k \frac{B_j}{\varphi^2 v_{k+1}}$. We consider the following two cases when $v_{k+1}$ increases further:
\begin{itemize}
\item $\frac{1}{\varphi}\sum_{j=1}^{k+1}B_j \leq v_k$. In this case, the value of $v_{k+1}$ first reaches $\frac{1}{\varphi}\sum_{j=1}^{k+1}B_j$ when increasing and her/his allocation is updated to
    $\frac{B_{k+1}}{\sum_{j=1}^{k+1}B_j}\min\{\hat{v}_{k+1}, 1\} = \frac{1}{\varphi} \cdot \frac{B_{k+1}}{\sum_{j=1}^{k+1}B_j} = \frac{1}{\varphi} - \frac{1}{\varphi^2} \sum_{j=1}^k \frac{B_j}{v_{k+1}}$.
   After that, this becomes an instance of case I and the allocation continues to increase as $v_{k+1}$ increases.
\item $\frac{1}{\varphi}\sum_{j=1}^{k+1}B_j > v_k$. In this case, the value of $v_{k+1}$ first reaches $v_k$ when increasing, and displace player $k$ to be the $k$th highest value. Then one of the following things will happen:
\begin{itemize}

\item If $\frac{1}{\varphi}(\sum_{j=1}^{k-1}B_j + B_{k+1}) > v_k$, since $\frac{1}{\varphi}\sum_{j=1}^{k-1}B_j < \frac{1}{\varphi}\sum_{j=1}^k B_j < v_k$, it is still an instance of case II and this agent $k+1$ is still the last agent in the winner set. The only difference is that the agent $k$ is not longer in the winner set and therefore the allocation $x_{k+1}$ gets updated to $\frac{1}{\varphi}-\frac{1}{\varphi^2}\sum_{j=1}^{k-1}\frac{B_j}{v_k} \geq \frac{1}{\varphi} - \frac{1}{\varphi^2}\sum_{j=1}^k\frac{B_j}{v_k}$.

\item If $\frac{1}{\varphi} (\sum_{j=1}^{k-1}B_j + B_{k+1}) \leq v_k$,  then it is still an instance of case II but with agent $k$ as the last agent in the winner set.  Agent $k+1$ become the second-to-last agent in the winner set and the allocation  $x_{k+1}$ gets updated to $\frac{ B_{k+1}}{\varphi v_k}\min\{\hat{v}_{k+1},1\} \geq \frac{ B_{k+1}}{\varphi^2v_k}\geq \frac{1}{\varphi} - \frac{1}{\varphi^2} \sum_{j=1}^k \frac{B_j}{v_k}$.
\end{itemize}
In both cases, the allocation is non-decreasing.
\end{itemize}
This concludes the proof of truthfulness of our auction.
\end{proof}


\begin{lemma}(Budget feasibility)
The payments defined in Auction~\ref{GSUPA} do not exceed the budgets.
\end{lemma}
\begin{proof}
For agents $i > k$ in case I and agents $i > k+1$ in case II,   this is trivial since they do not get any good and pay nothing.  For first $k$ agents in both case I and case II, their allocation do not change when they increase their valuations beyond $v_i \geq p_0$. In other words, their allocation is a constant when $v_i \geq p_0$. Thus, their payments are bounded by $p_0 x_i = p_0 \frac{B_i}{p_0}\min\{\hat{v}_i,1\} \leq B_i$.
The only remaining case is the $(k+1)$-th agent in case II. The payment is bounded by $v_{k+1}x_{k+1} = v_{k+1}( \frac{1}{\varphi} - \sum_{i = 1}^k \frac{ B_i}{\varphi^2 v_{k+1}}) \leq \frac{1}{\varphi} B_{k+1}$, where the inequality derives from the definition of $k$. This completes the proof.
\end{proof}

\subsection{Approximation Ratio Analysis}
Before we prove the approximation ratio, we obtain some bounds for the optimal liquid welfare. We refer to the optimal liquid welfare as $OPT = \max_{\mathbf{x}}\overline{W}(\mathbf{x})$. If we know all the information, the optimal can be computed by a simple greedy.

\begin{lemma}(\cite{DBLP:conf/icalp/DobzinskiL14})
\label{lemeasy}
The optimal liquid welfare $OPT$ occurs at $\bar{x}^\ast_i = \min(\frac{B_i}{v_i}, [1- \sum_{j<i} \bar{x}^\ast_j]^+)$, where $[x]^+=\max(0,x)$.
\end{lemma}

From this, it is easy to verify that the following expression for any $j\in [n-1]$ gives upper bounds for $OPT$, which holds even if $1-\sum_{i=1}^j \frac{B_i}{v_i}<0$:
\begin{equation}\label{bounds:OPT}
 OPT\leq \sum_{i=1}^j B_i + v_{j+1}(1-\sum_{i=1}^j \frac{B_i}{v_{i}}).
 \end{equation}
We propose our analysis of approximation ratio by the following lemma:
\begin{lemma}
\label{thmratio}
The liquid welfare achieved by Auction~\ref{GSUPA} is at least  $\frac{1}{\varphi }\cdot OPT$.
\end{lemma}

\begin{proof}
We prove for case I first. For $i \leq k_1$, since $v_i \ge p_0$, we have $v_i x_i = \frac{B_i v_i}{\sum_{j=1}^k B_j} \geq B_i$. For $k_1 < i \leq k$ we have $v_i x_i = \frac{B_i v_i}{\sum_{j=1}^k B_j} \min\{\hat{v}_i, 1\} \geq \hat{v}^2_i B_i$ as $\hat{v}_i\in [\frac{1}{\varphi}, 1)$ for $k_1 < i \leq k$. Thus

$$
\overline{W}(\mathbf{x}) = \sum_{i=1}^{k_1} B_i + \sum_{i=k_1+1}^k B_i \hat{v}^2_i.
$$

For optimal liquid welfare, we shall prove that
$$OPT\leq \sum_{i=1}^{k}B_i + (1-\sum_{i=k_1+1}^k \frac{B_i}{v_i})\frac{1}{\varphi} p_0 = \varphi \sum_{i=1}^{k_1}B_i + \sum_{i=k_1+1}^k B_i (\varphi  - \frac{1}{\varphi \hat{v}_i}).
$$

The equality part is by substituting $p_0=\sum_{i=1}^kB_i$, $\hat{v}_i = \frac{v_i}{\sum_{j=1}^k B_j}$ and  direct calculation. We prove the inequality by a case analysis in the following.
\begin{itemize}
\item $\sum_{i=k_1+1}^k \frac{B_i}{v_i} \leq \sum_{i=1}^k \frac{B_i}{v_i} \leq 1$. We use the bound (\ref{bounds:OPT}) for $OPT$ with $j=k$ and the fact that $v_{k+1}\leq \frac{1}{\varphi} p_0 $:
    \begin{align*}
     OPT  \leq \sum_{i=1}^k B_i + v_{k+1}(1-\sum_{i=1}^k \frac{B_i}{v_i})
     &\leq \sum_{i=1}^k B_i + v_{k+1}(1-\sum_{i=k_1+1}^k \frac{B_i}{v_i})   \\
     &  \leq \sum_{i=1}^{k}B_i +  (1-\sum_{i=k_1+1}^k \frac{B_i}{v_i})\frac{1}{\varphi} p_0
%
\end{align*}
%
%

\item $\sum_{i=k_1+1}^k \frac{B_i}{v_i} \leq 1< \sum_{i=1}^k \frac{B_i}{v_i}$. Then in the optimal solution, first $k$ agents are  not fully occupied(which means in the optimal solution, agent $k$'s budget is not used up). Thus nothing is allocated for agents $i \geq k+1$. In this case, we have $OPT \leq \sum_{i = 1}^k B_i \leq \sum_{i = 1}^k B_i + (1-\sum_{i=k_1+1}^k \frac{B_i}{v_i})\frac{1}{\varphi} p_0$ as the last term is non-negative.

\item $1<\sum_{i=k_1+1}^k \frac{B_i}{v_i} \leq \sum_{i=1}^k \frac{B_i}{v_i}$. We use the bound (\ref{bounds:OPT}) for $OPT$ with $j=k-1$:
  \begin{align*}
OPT  \leq \sum_{i=1}^{k-1} B_i + (1-\sum_{i=1}^{k-1} \frac{B_i}{v_i})v_k&= \sum_{i=1}^{k}  B_i  + (1-\sum_{i=1}^{k} \frac{B_i}{v_i})v_k \\&\leq \sum_{i=1}^{k} B_i  + (1-\sum_{i=k_1+1}^{k} \frac{B_i}{v_i})v_k \\& \leq \sum_{i=1}^{k}  B_i  + (1-\sum_{i=k_1+1}^{k} \frac{B_i}{v_i})\frac{1}{\varphi} p_0.
\end{align*}
     The last inequality uses the fact that $(1-\sum_{i=k_1+1}^{k} \frac{B_i}{v_i})<0$ and $v_k\geq \frac{1}{\varphi} p_0$.
\end{itemize}

To bound the liquid welfare of our auction, we need to give a good bound for $\hat{v}^2_i$ for agents $i\in [k_1+1, k]$. Noticing that for these agents, $\hat{v}_i \in [\frac{1}{\varphi}, 1]$, we shall prove that $\hat{v}^2_i \geq \frac{1}{\varphi} (\varphi-\frac{1}{\varphi \hat{v}_i})$.
To prove this, consider the following function
$$f(t) = \frac{\varphi  - \frac{1}{\varphi t}}{t^2}, t\in [\frac{1}{\varphi} ,1].$$
The derivative of $f(t)$ is
$$f'(t) = -2\frac{\varphi }{t^3}+3\frac{1}{\varphi t^4} = \frac{1}{t^4}[\frac{3}{\varphi} - 2\varphi t] \leq 0 \ \mbox{ when } \ t\in [\frac{1}{\varphi}, 1].$$
So $f(t)$ is monotone decreasing in interval $[\frac{1}{\varphi}, 1]$, and $f_{\max} = f(\frac{1}{\varphi}) = \varphi $.

By the property of $f$ and the fact that $\hat{v}_i\in [\frac{1}{\varphi}, 1)$ for $i=k_1+1, k_1+2,\ldots, k$, it is obvious that $\forall i \in [k_1+1, k]$, $ \hat{v}^2_i\geq \frac{1}{\varphi} (\varphi-\frac{1}{\varphi \hat{v}_i})$. Thus
\begin{align*}
\overline{W}(\mathbf{x}) \geq \sum_{i=1}^{k_1}B_i +  \sum_{i=k_1+1}^k B_i\hat{v}_i^2
&\geq \sum_{i=1}^{k_1}B_i + \sum_{i=k_1+1}^k B_i \frac{1}{\varphi}(\varphi  - \frac{1}{\varphi \hat{v}_i})\\
&=\frac{1}{\varphi}(\varphi \sum_{i=1}^{k_1}B_i + \sum_{i=k_1+1}^k B_i (\varphi  - \frac{1}{\varphi \hat{v}_i})) \\
&\geq \frac{1}{\varphi} OPT.
\end{align*}

This completes the proof for instances of case I and now we prove for instances of case II.  For $i \leq k_1$ since $v_i \geq \varphi v_{k+1}$ we have
$v_i x_i = v_i\frac{  B_i}{\varphi v_{k+1}}\min\{\hat{v}_i,1\}\geq B_i$.
For $k_1 < i \leq k$ we have
$v_i x_i = v_i\frac{  B_i}{\varphi v_{k+1}}\min\{\hat{v}_i,1\} = \frac{\hat{v}_i B_i }{\varphi v_{k+1}} \hat{v}_i\varphi  v_{k+1} = \hat{v}^2_i B_i<B_i$.
For $i = k+1$ we have
$$v_{k+1}x_{k+1} = v_{k+1}( \frac{1}{\varphi} - \sum_{i = 1}^k \frac{ B_i}{\varphi^2 v_{k+1}}) \leq \frac{1}{\varphi^2} B_{k+1}<B_{k+1}.$$
where the first inequality derives from the definition of $k$.

Thus, we can bound the liquid welfare as follows
\begin{align*}
\overline{W}(\mathbf{x})&= \sum_{i=1}^{k_1}B_i + \sum_{i=k_1+1}^k B_i\hat{v}^2_i+ v_{k+1}x_{k+1}\\
&= \sum_{i=1}^{k_1}B_i+ \sum_{i=k_1+1}^k B_i\hat{v}^2_i + \frac{1}{\varphi} v_{k+1} - \frac{1}{\varphi^2} \sum_{i=1}^{k_1}B_i - \frac{1}{\varphi^2}\sum_{i=k_1+1}^k B_i \\
&= (1-\frac{1}{\varphi^2}) \sum_{i=1}^{k_1}B_i  + \frac{1}{\varphi} v_{k+1} + \sum_{i=k_1+1}^k B_i (\hat{v_i}^2 - \frac{1}{\varphi^2}).
\end{align*}

For optimal liquid welfare, we use our bound (\ref{bounds:OPT}) with $j=k$ to get
$$OPT \leq \sum_{i=1}^{k}B_i+ (1-\sum_{i=k_1+1}^k \frac{B_i}{v_i})v_{k+1} = \sum_{i=1}^{k_1}B_i + v_{k+1} + \sum_{i = k_1+1}^k B_i(1-\frac{1}{\varphi\hat{v}_i}),
$$
where the equality part is by substituting $\hat{v}_i = \frac{ v_i}{\varphi v_{k+1}}$ and  direct calculation.
%
%
%
To bound the liquid welfare of our auction, we need to give a good bound for $\hat{v}^2_i$ for agents $i\in [k_1+1, k]$. Noticing that for these agents, $\hat{v}_i \in [\frac{1}{\varphi}, 1]$, we shall prove that $\hat{v}^2_i - \frac{1}{\varphi^2} \geq \frac{1}{\varphi}(1-\frac{1}{\phi \hat{v}_i})$.
To prove this, consider the following function

$$g(t) = (t^2-\frac{1}{\varphi^2})/(1-\frac{1}{\varphi t})= (t+\frac{1}{\varphi})t, t\in[\frac{1}{\varphi},1].$$
 It is clear that $g$ is monotone increasing on $t$ in the interval $[\frac{1}{\varphi},1]$, so $g_{\min} = g(\frac{1}{\varphi}) = 2\frac{1}{\varphi^2} > \frac{1}{\varphi}$.

By the property of $g$ and the fact that $\hat{v}_i\in [\frac{1}{\varphi}, 1)$ for $i=k_1+1, k_1+2,\ldots, k$, it is obvious that $\forall i \in [k_1+1, k]$, $ \hat{v}^2_i - \frac{1}{\varphi^2} \geq \frac{1}{\varphi}(1-\frac{1}{\phi \hat{v}_i})$. Thus
\begin{align*}
\overline{W}(\mathbf{x})&= (1-\frac{1}{\varphi^2})\sum_{i=1}^{k_1}B_i + \frac{1}{\varphi} v_{k+1} + \sum_{i=k_1+1}^k B_i (\hat{v}^2_i-\frac{1}{\varphi^2}) \\
&> \frac{1}{\varphi} \sum_{i=1}^{k_1}B_i + \frac{1}{\varphi} v_{k+1} + \frac{1}{\varphi}\sum_{i=k_1+1}^k B_i (1-\frac{1}{\varphi \hat{v}_i}) \\
& \geq \frac{1}{\varphi} OPT.
\end{align*}
This completes the proof for approximation ratio.
\end{proof}
To conclude this section, here we provide the following example showing that the analysis of our auction is tight.
\begin{exam}(Tightness)
Consider two agents with profiles $v_1 = 1$, $B_1 = \epsilon$ and $v_2 = \frac{1}{\varphi}$, $B_2 = 1-\epsilon$ where $\epsilon \in (0,1)$. It is easy to verify that $OPT = \epsilon + (1-\epsilon)\frac{1}{\varphi}$. For Auction \ref{GSUPA}, $\overline{W}(\mathbf{x})= B_1 + B_2\frac{1}{\varphi^2} = \epsilon + (1-\epsilon)\frac{1}{\varphi^2}$. When $\epsilon \rightarrow 0$, $\overline{W}(\mathbf{x})= \frac{1}{\varphi} OPT$
\end{exam}

\section{Private Budget}

In this section, we deal with the setting that agents' budgets are private information that the
auctioneer must design a mechanism which incentives agents to report their true values and budgets. We also study the general case that
 the valuation function $v_i(\cdot)$ for each agent $i$ could be any monotone non-decreasing function.

For a subset of agents $Q \subseteq [n]$, let $OPT(Q)$ denote the optimal liquid welfare for agents in group $Q$. Formally $OPT(Q) = \max_{x} \sum_{i\in Q} \min\{v_i(x_i), B_i\}$. In particular, let $OPT = OPT([n])$ which is our objective in this setting.

Our new auction for private budget model is presented in Auction~\ref{RS}, where the parameters
$\gamma>1, 0<\beta<\frac{1}{2}$ and $0\leq \mu \leq 1$ shall be specified later.

Basically, it is a combination of the following three basic auctions:
\begin{itemize}
  \item With probability of $\frac{\mu}{3}$, we run the first modified Vickrey auction.  Agent $i_1$ with highest  $\bar{v}_i=\min\{v_i(1), B_i\}$ gets the total unit of the good and needs to pay $p_{i_1} = \gamma\bar{v}_{i_2}>\bar{v}_{i_2}$ , which is strictly higher than the second highest $\bar{v}_i$. If agent $i_1$ is not willing to pay ($v_{i_1}<\gamma\bar{v}_{i_2}$) or does not have enough budget to pay ($B_{i_1}<\gamma\bar{v}_{i_2}$), we simply refuse to sell the item to any one.
    \item With probability of $\frac{2\mu}{3}$, we run the second modified Vickrey auction.  Agents are randomly divided into two groups $S$ and $T$. We only sell the total unit of the good to the first agent (with a prior fixed order) in group $S$ who is willing and able to pay the price $\frac{\max_{i\in T} \bar{v}_i}{\gamma}$. If there is no such agent in group $S$, we simply refuse to sell the item to any one.
  \item With the remaining probability of $1-\mu$, we run a random sampling auction.  Agents are randomly divided into two groups $S$ and $T$. We sell half of the good to agents in group $S$ with fixed price $\beta OPT(T)$ per unit. More precisely,
  for each agent in group $S$ with a prior fixed order, we simply offer a price $\beta OPT(T)$ per unit and let the agent get the most profitable fraction of the good within the availability of the good and budget of the agent.   This is precisely captured by the expression
  \[x_i \leftarrow \argmax \limits_{x\leq \min\{ x_S, \frac{B_i}{\beta OPT(T)}\} } \{v_i(x)-\beta OPT(T) x\}.\]
  If there are multiple $x$ that achieve the maximum, we choose the largest one.
  We do the same thing for agents in $T$ but with price  $\beta OPT(S)$.
\end{itemize}
We call the combination of the first two auctions the modified Vickrey auction and the third part as the random sampling auction.

%
%

\IncMargin{1em}
\begin{algorithm}[h]
\SetKwInOut{Input}{input}\SetKwInOut{Output}{output}
\Input{$n $ agents with values $v_1, \ldots,  v_n$ and budgets $B_1, \ldots, B_n$}
\Output{An allocation $(x_1, \ldots, x_n)$ and corresponding payments $(p_1, \ldots, p_n)$\;}
\BlankLine
\Begin{
\For{$i = 1$ \emph{\KwTo} $n$}{
		$x_i \leftarrow 0$,		$p_i \leftarrow 0$, $\bar{v}_i \leftarrow \min\{v_i(1), B_i\}$\;
	}
}
With probability of $\frac{\mu}{3}$
\Begin{
	$i_1 \leftarrow \arg \max_i \bar{v}_i $,	 $i_2 \leftarrow \arg \max_{i\ne i_1} \bar{v}_i$\;
	\If{$\bar{v}_{i_1} \geq \gamma \bar{v}_{i_2}$}{
		$x_{i_1} \leftarrow 1$,	$p_{i_1} \leftarrow \gamma\bar{v}_{i_2}$
	}
}
With probability of $\frac{2\mu}{3}$
\Begin
{
Randomly divide all agents with equal probability into set $S$ and $T$\;
$\bar{v}_T \leftarrow \max_{i\in T} \bar{v}_i$\;
    \ForAll{$i \in S$}{
        \If{$\bar{v}_i \geq \frac{\bar{v}_T}{\gamma}$ }
        {
            $x_i \leftarrow 1, p_i \leftarrow \frac{\bar{v}_T}{\gamma}$\;
            Halt
        }
    }
}
With  probability of $1-\mu$
\Begin{
Randomly divide all agents with equal probability into set $S$ and $T$,
$x_S \leftarrow \frac{1}{2}$,
$x_T \leftarrow \frac{1}{2}$\;

	\ForAll{$i \in S$}{
		$x_i \leftarrow \argmax \limits_{x\leq \min\{ x_S, \frac{B_i}{\beta OPT(T)}\} } \{v_i(x)-\beta OPT(T) x\}$\;
		$p_i \leftarrow  \beta OPT(T) x_i$\;
		$x_S \leftarrow  x_S - x_i$		
	}
	\ForAll{$i \in T$}{
		$x_i \leftarrow \argmax \limits_{x\leq \min\{ x_T, \frac{B_i}{\beta OPT(S)}\} } \{v_i(x)-\beta OPT(S) x \}$\;
		$p_i \leftarrow \beta OPT(S) x_i$\;
		$x_T \leftarrow x_T - x_i$		
	}

}

\caption{Random Sampling Auction for Private Budgets}\label{RS}
\end{algorithm}
\DecMargin{1em}

\begin{theorem}
\label{thmprimain}
Choosing $\beta =\frac{3}{10}, \gamma=\sqrt{\frac{10}{9}}$ and $\mu=\frac{5}{7}$, Auction~\ref{RS} is a truthful, budget feasible mechanism which guarantees liquid welfare of at least $\frac{1}{34}OPT$.
\end{theorem}

\subsection{Truthfulness and Budget Feasibility}
Before we prove that our auction is truthful, we first point out that the ordinarily Vickrey auction (i.e. $\gamma=1$) on $\bar{v}_i=\min \{v_i(1), B_i\}$ is not truthful. Here is an example in which the valuation function is additive (thus we use $v_i$ to denote price per unit to illustrate): two agents with profiles $(v_1,B_1)=(v_2,B_2)=(2,1)$. If both of the agents bid truthfully, whatever the tie-breaking rule the Vickrey auction uses (even if we allow randomness), at least for one of the agents, the probability she/he gets the total unit of good is strictly less than $1$. For symmetry, we assume that the probability agent $1$ gets the total unit of good is strictly less than $1$. When agent $1$ does get some fraction of the good, she/he needs to pay $1$ per unit according to the Vickrey's rule. As a result, the expected utility of agent $1$ is strictly less than $2-1=1$. However, if agent $1$ bids $(v'_1,B'_1)=(2,1.5)$, she/he will get the total unit of the good for sure based on Vickrey auction and the payment is still $1$, which does not exceed the budget. Therefore, her/his utility become $1$, which is strictly better than bidding truthfully.

This is the reason why we need to modify the Vickrey auction. In the following, we prove that Auction~\ref{RS} in which our modification is applied, is universally truthful.

\begin{lemma}
Auction~\ref{RS} is universally truthful.
\end{lemma}
\begin{proof}
The auction is a probabilistic combination of three auctions. For the second and third auctions, it also uses random bits to do the partition of $(S,T)$.  We only need to prove that all of them are truthful when these partitions are fixed.

For the first modified Vickrey auction on value $\bar{v}_i=\min \{v_i(1), B_i\}$, two cases may happen:
\begin{itemize}
\item $\bar{v}_{i_1} \geq \gamma \bar{v}_{i_2}$. In this case, for any agent $j$ other than $i_1$, $\bar{v}_j \leq \bar{v}_{i_1}< \bar{v}_{i_1}\gamma $. If $j$ wants to change her/his bid to become the winner, she/he needs to pay $\bar{v}_{i_1}\gamma$, which is strictly greater than either her/his true budget or true value. This leads to not enough budget or negative utility. For $i_1$, she/he does not have the incentive to change value or budget, since the payment is decided by $i_2$, which is no larger than her/his value and also within her/his budget.
\item $\bar{v}_{i_1} < \gamma \bar{v}_{i_2}$. In this case, no one is the winner. For any agent $j$ other than $i_1$, it is the same argument as before.  For $i_1$, if she/he wants to change her/his bid to become the winner, then she/he needs to pay $\gamma \bar{v}_{i_2}$, which is strictly greater than either her/his true budget or true value. This leads to not enough budget or negative utility.
\end{itemize}

Now we prove for the second modified Vickrey auction. For any agent in $T$, she/he does not get any fraction of the good regardingless of her/his bid. So, they do not have incentive to lie.  For any agent in $S$,
she/he cannot change the price per unit or her/his position in the order by changing her/his bid. When an agent in $S$ has chance to get the good, it is simply a take-it-or-leave-it offer
 with fixed price. So, they do not have incentive to lie.

For the random sampling auction part, each agent cannot change her/his price per unit or position by changing her/his bid, and given a fixed price and position, a agent has already got the most profitable fraction of the good. Therefore, agents do not have incentive to change their bids.

Thus, all the three auctions above are truthful. This concludes the proof for universally truthfulness.
\end{proof}

\begin{lemma}
Auction~\ref{RS} is budget feasible.
\end{lemma}
\begin{proof}
For the first modified Vickrey auction, if no one wins, everyone's payment is zero; if $i_1$ wins, then she/he pays $\gamma \bar{v}_{i_2}\leq \bar{v}_{i_1}= \min \{v_{i_1}(1), B_{i_1}\}\leq B_{i_1}$.
For the second modified Vickrey auction, if no one wins, everyone's payment is zero; if $i^\ast$ wins, then she/he pays $\gamma \frac{\bar{v}_T}{\gamma}\leq \bar{v}_{i^\ast}= \min \{v_{i^\ast}(1), B_{i^\ast}\} \leq B_{i^\ast}$.

For random sampling auction and agent $i\in S$,
			$p_i = x_i \cdot \beta OPT(T) \leq \min\{ x_S, \frac{B_i}{\beta OPT(T)}\} \cdot \beta OPT(T)  \leq \frac{B_i}{\beta OPT(T)} \cdot \beta OPT(T) =B_i. $
Same thing also holds for agents in $T$.

This concludes the proof for budget feasibility.
\end{proof}

\subsection{Approximation Ratio Analysis}
We first prove the following lemma, which bounds the liquid welfare of the auction by its revenue. This is useful in our analysis.
\begin{lemma}
\label{l1}
Liquid welfare  produced by any truthful and budget feasible mechanism is at least the revenue of the auctioneer.
\end{lemma}

\begin{proof}
For an allocation $\mathbf{x} = (x_1, x_2, \ldots, x_n)$ and payment $\mathbf{p} = (p_1, p_2, \ldots, p_n)$ given by such a mechanism,  we have $v_i (x_i) \geq p_i$ by truthfulness and $B_i \geq p_i$ by budget feasibility. So the liquid welfare $\overline{W}(x) = \sum_i \min\{v_i (x_i), B_i\} \geq \sum_i p_i$.
\end{proof}

Based on Lemma~\ref{l1}, we prove that the modified Vickrey auction part performs well when $\max_i \bar{v}_i$ is large.
\begin{lemma}\label{bound_vickrey}
Let $\gamma=\sqrt{\frac{10}{9}}$. Then the modified  Vickrey auction part get expect liquid welfare of $\frac{3\mu}{10} \max_i \bar{v}_i$.
\end{lemma}
We note that by choosing $\gamma$ arbitrarily close to $1$, we can get liquid welfare arbitrarily close to $\frac{\mu}{3} \max_i \bar{v}_i$.
We choose the above value for the notational simplicity of the presentation.
\begin{proof}
If the highest two $\bar{v}_i$ are not relatively close to each other, namely $\bar{v}_{i_1} \geq \gamma \bar{v}_{i_2}$ ($i_1, i_2$ are as defined in Auction~\ref{RS}). Then the first modified Vickrey auction successes and gets expect liquid welfare of $\frac{\mu}{3} \max_i \bar{v}_i >\frac{3\mu}{10} \max_i \bar{v}_i$.

If the highest two $\bar{v}_i$ are relatively close to each other, namely $\bar{v}_{i_1} < \gamma \bar{v}_{i_2}$. Then in the second modified Vickrey auction, with probability $\frac{1}{2}$ these two agents $i_1, i_2$  are put into different groups. When this event occurs, the second modified Vickrey successes and gains a revenue of at least $\frac{\max_i \bar{v}_i}{\gamma^2}$. According to Lemma~\ref{l1}, liquid welfare extract by the second modified auction is greater than this revenue, thus contributes at least $\frac{\max_i \bar{v}_i}{\gamma^2}$ liquid welfare. Therefore, the  expect liquid welfare in this case is at least $\frac{2\mu}{3} \frac{1}{2} \frac{\max_i \bar{v}_i}{\gamma^2} =\frac{3\mu}{10} \max_i \bar{v}_i$.
\end{proof}

If $\max_i \bar{v}_i$ is already a significant fraction of the optimal solution, we are already done. In the following, we shall prove that the random sampling auction performs well when $\max_i \bar{v}_i$ is small.
 We first give some definitions. Let $\bar{v}_i(x)=\min\{v_i(x), B_i\}$ be the capped valuation for agent $i$.
 Then, for any allocation $\mathbf{x}=(x_1, \ldots, x_n)$ , we have $\overline{W}(\mathbf{x})=\sum_i \bar{v}_i(x_i)$. The following notion plays an important role in our analysis. We define
 \[D_i(p) = \argmax_{x\leq 1} \{\bar{v}_i(x)-x p \}. \]
 If there are multiple $x$ that achieve the maximum, we choose the largest one. It is very crucial that we use $\bar{v}_i(x)$ rather than $v_i(x)$ in the definition of $D_i(p)$. By this definition, we can directly see that for any $p>0$ and $x<D_i(p)$, we have $D_i(p)\leq \frac{B_i}{p}$ and  $v_i(x)< B_i$.  This $D_i(p)$ also gives a lower bound of agent $i$'s demand if there are enough availability of the good. Formally, we have the following lemma.
\begin{lemma}\label{lemma:demand}
Let $p > 0$ and $D_i(p) \leq X$. Then
\[\argmax \limits_{x\leq \min\{ X, \frac{B_i}{p}\} } \{v(x)-xp \} \ge D_i(p) .\]
\end{lemma}
 \begin{proof}
 The left hand side of the inequality is agent $i$'s most profitable fraction given price $p$ per unit and the total availability of good of $X$. For $x<D_i(p) \leq \min( X, \frac{B_i}{p}) $, we have
 \[ v(x)-xp = \bar{v}_i(x)-x p \leq \bar{v}_i(D_i(p))-D_i(p) p \leq v_i(D_i(p))-D_i(p) p.\]
 The first equality uses the fact that $v_i(x)< B_i$ for  $x<D_i(p)$; the first inequality uses the definition of  $D_i(p)$; and the last inequality uses the fact that $\bar{v}_i(x)\leq v_i(x)$ for any $x$.
 Since we always break ties in favor of larger $x$, the maximum of  $v(x)-xp$ in the LHS is archived by $x\ge D_i(p)$. This completes the proof.
 \end{proof}

Let $W(p) = \sum_i \bar{v}_i(D_i(p))$. The intuition for this notion is that with fixed price $p$, $W(p)$ gives the maximum liquid welfare from all agents. We present the following lemma giving an lower bound for this notion using $OPT$ and fixed price $p$.

 \begin{lemma}\label{lw}
 For any $p\geq 0$, $W(p) \geq OPT-p$.
 \end{lemma}
 \begin{proof}
 Let $(x_1, \ldots, x_n)$ be an instance of optimal allocation. By the definition of $D_i(p)$, we have
$\bar{v}_i(D_i(p))-D_i(p) p \ge \bar{v}_i(x_i)-x_i p$ ,
 and thus $\bar{v}_i(D_i(p)) \ge \bar{v}_i(x_i)-x_i p$.
Sum up these inequalities for all $i\in [n]$, we get
 \[ \sum_i \bar{v}_i(D_i(p)) \ge \sum_i  (\bar{v}_i(x_i)-x_i p) =OPT -p \sum_i x_i \ge  OPT-p.\]
 \end{proof}

The following facts on relationship between $OPT(S)$, $OPT(S)$ and $OPT$, are obvious.

\begin{lemma}
\label{l2}
Let $(S,T)$ be a partition of $[n]$. Then
$OPT(S)\leq OPT$, $OPT(T)\leq OPT$ and $OPT(S)+OPT(T)\geq OPT$.
\end{lemma}

By choosing $p=\beta OPT$ in Lemma \ref{lw}, we get that
\[W(\beta OPT) = \sum_i \bar{v}_i(D_i(\beta OPT)) \geq OPT-\beta OPT=(1-\beta) OPT.\]
  This is a constant fraction of $OPT$. Since $OPT(S), OPT(T)\leq OPT$, the fraction of good demanded by agent $i$ in the random sampling auction is at least $D_i(\beta OPT)$ by Lemma \ref{lemma:demand} providing that there are enough fraction of the good remains.
This is a good approximation of the optimal liquid welfare when each of $\bar{v}_i(D_i(\beta OPT))$ is small. Let $W = W(\beta OPT)$, $W_S=\sum_{i\in S} \bar{v}_i(D_i(\beta OPT))$ and $W_T=\sum_{i\in T} \bar{v}_i(D_i(\beta OPT))$.
 We first prove that both sets $S$ and $T$ get significant amount of demands at fixed price $\beta OPT$ with high probability in this case.

\begin{lemma}\label{lemma:bound}
If $\max_{i\in [n]} \bar{v}_i(D_i(\beta OPT)) \leq \alpha \cdot OPT $, then
$$Pr(W_S, W_T \ge \frac{\beta}{2} OPT ) \geq 1-\frac{\alpha (1-\beta)}{(1-2 \beta )^2}.$$
\end{lemma}

\begin{proof}
Let $I_i$ to the random indicator variable for the event $i\in S$. Then $W_S = \sum_{i\in S}\bar{v}_i(D_i(\beta OPT))= \sum_{i\in [n]} \bar{v}_i(D_i(\beta OPT)) I_i$, $\be(W_S)=\frac{1}{2} W$ and
\begin{align*}
\begin{split}
\bv(W_S) = \sum_{i\in S} \bv( \bar{v}_i(D_i(\beta OPT)) I_i) &= \sum_{i\in S} (\be((\bar{v}_i(D_i(\beta OPT)) I_i)^2) - \be(\bar{v}_i(D_i(\beta OPT)) I_i)^2)\\
&= \sum_{i\in S} \frac{1}{4} (\bar{v}_i(D_i(\beta OPT)))^2 \\
&\leq	 \frac{W}{\alpha OPT}\cdot \frac{1}{4} \cdot(\alpha OPT)^2 \\
&= \frac{1}{4} \alpha W OPT,
\end{split}
\end{align*}
where the inequality uses the fact that $\max_{i\in V} \bar{v}_i(D_i(\beta OPT)) \leq \alpha OPT$.

By Chebyshev's Inequality, we have:
\begin{align*}
Pr(\frac{\beta}{2} OPT \leq W_S \leq W-\frac{\beta}{2} OPT) & = Pr( |W_S - \be(W_S) | \leq \frac{W-\beta OPT}{2})  \\
&\geq 1- \frac{\bv(W_S)}{(\frac{W-\beta OPT}{2})^2}  \\
&\geq 1-\frac{\alpha W OPT}{(W-\beta OPT)^2}\\
&\geq 1-\frac{\alpha (1-\beta)}{(1-2 \beta )^2},
\end{align*}
where the last inequality uses the fact that $W\ge (1-\beta)OPT$.
Since $W_S+W_T=W$, the event of $(\frac{\beta}{2} OPT \leq W_S \leq W-\frac{\beta}{2} OPT)$ is the same as the event of $(W_S, W_T \ge \frac{\beta}{2} opt)$. This completes the proof.
\end{proof}

The following lemma gives a bound of liquid welfare for the random sampling auction part.
\begin{lemma}\label{bound_sampling}
If $\max_i \bar{v}_i =\alpha \cdot OPT$, then the random sampling auction part gets at least expected liquid welfare of  $(1-\mu) (1-\frac{\alpha (1-\beta)}{(1-2 \beta )^2}) (\frac{1}{2}-\frac{\alpha}{\beta})\beta OPT$.
\end{lemma}

\begin{proof}
Since $\max_{i\in [n]} \bar{v}_i(D_i(\beta OPT)) \leq \max_i \bar{v}_i =\alpha \cdot OPT$, by Lemma \ref{lemma:bound} we know that $Pr(W_S, W_T \ge \frac{\beta}{2} OPT ) \geq 1-\frac{\alpha (1-\beta) }{(1-2\beta)^2}$ . We only bound the liquid welfare  when this good event $(W_S, W_T \ge \frac{\beta}{2} OPT) $ occurs, which occurs with probability $(1-\frac{\alpha (1-\beta)}{(1-2 \beta )^2})$.

Not only $ \bar{v}_i(D_i(\beta OPT))$ is bounded from above, $D_i(\beta OPT)$ is also bounded from above. From the definition of $D_i(\cdot)$, we know that $\bar{v}_i(D_i(\beta OPT)) - \beta OPT D_i(\beta OPT)\geq 0$. Therefore,
\[D_i(\beta OPT) \leq \frac{\bar{v}_i(D_i(\beta OPT))}{\beta OPT} \leq \frac{\alpha}{\beta}.\]
We first consider liquid welfare obtained by agents in $T$.
If every agent $i \in T$ gets at least $v_i(D_i(\beta OPT))$ fraction of good, then the total liquid welfare of our auction is at least $W_T\geq \frac{\beta}{2} opt $. Otherwise, due to Lemma \ref{lemma:demand}, it must be the case that there is not enough good remains. Since $D_i(\beta OPT) \leq \frac{\alpha}{\beta}$, we know that at least $\frac{1}{2}-\frac{\alpha}{\beta}$ fraction of the good is sold. This extracts a revenue of $(\frac{1}{2}-\frac{\alpha}{\beta})\beta OPT(S)$ and thus also liquid welfare of this amount by Lemma \ref{l1}. Put these two cases together, the liquid welfare for group $T$ in our auction is at least $\min\{\frac{\beta}{2} OPT , (\frac{1}{2}-\frac{\alpha}{\beta})\beta OPT(S)  \} = (\frac{1}{2}-\frac{\alpha}{\beta})\beta OPT(S)$.  By similar argument, the liquid welfare from agents in group $S$ is at least
$(\frac{1}{2}-\frac{\alpha}{\beta})\beta OPT(T)$.

To sum up, the total expected liquid welfare is at least
\[ (1-\mu) (1-\frac{\alpha (1-\beta)}{(1-2 \beta )^2})  (\frac{1}{2}-\frac{\alpha}{\beta})\beta ( OPT(S)+ OPT(T)) \geq (1-\mu) (1-\frac{\alpha (1-\beta)}{(1-2 \beta )^2}) (\frac{1}{2}-\frac{\alpha}{\beta})\beta OPT.\]
\end{proof}

Finally, we estimate the  approximation ratio of Auction~\ref{RS}.
\begin{lemma}
\label{thmRS}
Choosing $\beta=\frac{3}{10}, \gamma=\sqrt{\frac{10}{9}}$ and $\mu=\frac{5}{7}$, the approximation ratio of Auction~\ref{RS} is at most $34$.
\end{lemma}

\begin{proof}
Assume that $\max_i \bar{v}_i =\alpha \cdot OPT$. By Lemma \ref{bound_vickrey} and Lemma \ref{bound_sampling}, the total expected liquid welfare is at least
\[\left(\frac{3\mu}{10} \alpha + (1-\mu) (1-\frac{\alpha (1-\beta)}{(1-2 \beta )^2}) (\frac{1}{2}-\frac{\alpha}{\beta})\beta\right)OPT.\]
Substitute $\beta$, $\mu$ with the above specified value and simplify, the above expression is $(\frac{5}{4}\alpha^2 -\frac{29}{112}\alpha + \frac{3}{70})OPT$. One can easily check that the minimum of this expression is greater than $\frac{1}{34}OPT$, thus our auction has an approximation ratio of at most $34$.
\end{proof}

\subsection{Robustness of the Auction}
Our auction is rather robust in terms of the setting. We do not have any requirement about the valuation functions. Technically,  in the presentation we still use the assumption
that the valuation function is monotone. In most of the cases, this is true or without loss of generality since the agent can
simply discard certain amount of the good. Even if this is not the case, we can also easily modify the auction to be compatible with possible non-monotone valuation functions.  The only place we need to modify is that when the current auction assigns the total unit of the good to an agent, it assigns the most valuable fraction to her/him and discard the remaining.

For the simplicity of the presentation, we assume that the good is divisible. Our auction is also good if the items are not continuously divisible.  For example, the mechanism works for the multi-unit auction even if each unit of the good is not divisible.

Another issue has not been discussed is the computational complexity of the auction as we mainly focus on the approximation ratio caused by the truthfulness and budget feasibility constrain. The computational complexity depends on how to represent the input valuation functions. If these are linear valuations and each can be simply represented by a single number, our auction is indeed efficient.  If the valuation functions are given as generic value oracles, then it is even intractable to computable the most profitable fraction for an agent given a fixed price. So, a reasonable assumption is that valuation functions are given by demand oracles or  in some concise representation. Then the main problematic step is to compute the optimal solution for an off line instance.   This could be at least NP-hard even for some concise representations. For example, we can easily encode knapsack problem here. Then, another robustness of the auction is that it still works well when we replace the optimal solution with some constant approximation.
Therefore, as long as we can design an polynomial time algorithm with constant approximation ratio for the off line optimization problem, we can design an   auction, which is truthful, budget feasible, of constant approximation and polynomial time computable.

\bibliographystyle{plain}

\bibliography{game}

\begin{thebibliography}{10}

\bibitem{BeiCGL12}
Xiaohui Bei, Ning Chen, Nick Gravin, and Pinyan Lu.
\newblock Budget feasible mechanism design: from prior-free to bayesian.
\newblock In {\em STOC}, pages 449--458, 2012.

\bibitem{BorgsCIMS05}
Christian Borgs, Jennifer~T. Chayes, Nicole Immorlica, Mohammad Mahdian, and
  Amin Saberi.
\newblock Multi-unit auctions with budget-constrained bidders.
\newblock In {\em EC}, pages 44--51, 2005.

\bibitem{ChawlaMM11}
Shuchi Chawla, David~L. Malec, and Azarakhsh Malekian.
\newblock Bayesian mechanism design for budget-constrained agents.
\newblock In {\em EC}, pages 253--262, 2011.

\bibitem{ChenGL11}
Ning Chen, Nick Gravin, and Pinyan Lu.
\newblock On the approximability of budget feasible mechanisms.
\newblock In {\em SODA}, pages 685--699, 2011.

\bibitem{ChenGL13}
Ning Chen, Nick Gravin, and Pinyan Lu.
\newblock Truthful generalized assignments via stable matching.
\newblock {\em to appear in Mathematics of Operations Research}, 2013.

\bibitem{Clarke1971}
E.~H. Clarke.
\newblock Multipart pricing of public goods.
\newblock {\em Public Choice}, 11(1), September 1971.

\bibitem{DevanurHH13}
Nikhil~R. Devanur, Bach~Q. Ha, and Jason~D. Hartline.
\newblock Prior-free auctions for budgeted agents.
\newblock In {\em EC}, pages 287--304, 2013.

\bibitem{DobzinskiLN12}
Shahar Dobzinski, Ron Lavi, and Noam Nisan.
\newblock Multi-unit auctions with budget limits.
\newblock {\em Games and Economic Behavior}, 74(2):486--503, 2012.

\bibitem{DBLP:conf/icalp/DobzinskiL14}
Shahar Dobzinski and Renato~Paes Leme.
\newblock Efficiency guarantees in auctions with budgets.
\newblock In {\em ICALP (1)}, pages 392--404, 2014.

\bibitem{DobzinskiPS11}
Shahar Dobzinski, Christos~H. Papadimitriou, and Yaron Singer.
\newblock Mechanisms for complement-free procurement.
\newblock In {\em EC}, pages 273--282, 2011.

\bibitem{FeldmanFLS12}
Michal Feldman, Amos Fiat, Stefano Leonardi, and Piotr Sankowski.
\newblock Revenue maximizing envy-free multi-unit auctions with budgets.
\newblock In {\em EC}, pages 532--549, 2012.

\bibitem{FiatLSS11}
Amos Fiat, Stefano Leonardi, Jared Saia, and Piotr Sankowski.
\newblock Single valued combinatorial auctions with budgets.
\newblock In {\em EC}, pages 223--232, 2011.

\bibitem{Goldberg2006}
A.V. Goldberg, J.D. Hartline, A.R. Karlin, M.~Saks, and A.~Wright.
\newblock {Competitive auctions}.
\newblock {\em Games and Economic Behavior}, 55(2):242--269, 2006.

\bibitem{GravinL13}
Nick Gravin and Pinyan Lu.
\newblock Competitive auctions for markets with positive externalities.
\newblock In {\em ICALP (2)}, pages 569--580, 2013.

\bibitem{Groves1973}
T.~Groves.
\newblock Incentives in teams.
\newblock {\em Econometrica}, 41(4):617--631, July 1973.

\bibitem{myerson1981optimal}
Roger~B Myerson.
\newblock Optimal auction design.
\newblock {\em Mathematics of operations research}, 6(1):58--73, 1981.

\bibitem{Singer10}
Yaron Singer.
\newblock Budget feasible mechanisms.
\newblock In {\em FOCS}, pages 765--774, 2010.

\bibitem{SyrgkanisT13}
Vasilis Syrgkanis and {\'E}va Tardos.
\newblock Composable and efficient mechanisms.
\newblock In {\em STOC}, pages 211--220, 2013.

\bibitem{Vickrey1961}
W.~Vickrey.
\newblock Counterspeculation, auctions, and competitive sealed tenders.
\newblock {\em The Journal of Finance}, 16(1):8--37, 1961.

\end{thebibliography}



\end{document}